\documentclass[letterpaper, 10 pt, conference]{ieeeconf}  

\IEEEoverridecommandlockouts                              

\usepackage{amsmath,amssymb}
\usepackage{algorithm}
\usepackage{algpseudocode}
\usepackage{array}
\usepackage{subfigure}
\usepackage{textcomp}
\usepackage{stfloats}
\usepackage{url}
\usepackage{verbatim}
\usepackage{graphicx}
\usepackage{cite}
\usepackage{color}
\usepackage{pifont}
\usepackage{hyperref}
\usepackage{footnote}
\usepackage{diagbox}
\usepackage{booktabs}
\usepackage{multirow}
\usepackage{tablefootnote}
\hypersetup{
colorlinks=true,
linkcolor=black,
citecolor=black,
}

\graphicspath{{Figures/}}

\newtheorem{theorem}{Theorem}[section]
\newtheorem{proposition}[theorem]{Proposition}

\newcommand \ba{\mathbf{a}}
\newcommand \bb{\mathbf{b}}
\newcommand \bd{\mathbf{d}}
\newcommand \bff{\mathbf{f}}
\newcommand \bh{\mathbf{h}}

\newcommand \bp{\mathbf{p}}
\newcommand \bq{\mathbf{q}}
\newcommand \br{\mathbf{r}}

\newcommand \bv{\mathbf{v}}

\newcommand \bA{\mathbf{A}}

\newcommand \bI{\mathbf{I}}
\newcommand \bR{\mathbf{R}}
\newcommand \bX{\mathbf{X}}

\newcommand \mcD{\mathcal{D}}
\newcommand \mcE{\mathcal{E}}
\newcommand \mcG{\mathcal{G}}
\newcommand \mcL{\mathcal{L}}
\newcommand \mcN{\mathcal{N}}
\newcommand \0{\boldsymbol{0}}
\newcommand \1{\boldsymbol{1}}
\newcommand \bphi{\boldsymbol{\phi}}
\newcommand \btheta{\boldsymbol{\theta}}
\newcommand \blambda{\boldsymbol{\lambda}}

\newcommand \bnu{\boldsymbol{\nu}}
\newcommand \bzeta{\boldsymbol{\zeta}}
\newcommand \bbeta{\boldsymbol{\beta}}
\newcommand \bxi{\boldsymbol{\xi}}
\newcommand \bPhi{\boldsymbol{\Phi}}

\newcommand \diag{\textrm{diag}}

\def\realnumbers{\mathbb{R}}

\title{\LARGE \bf
Feedback Optimization of Incentives for Distribution Grid Services
}

\author{Guido Cavraro \and Joshua Comden \and Andrey Bernstein
\thanks{This work was authored by the National Renewable Energy Laboratory, operated by Alliance for Sustainable Energy, LLC, for the U.S. Department of Energy (DOE) under Contract No. DE-AC36-08GO28308. Funding provided by the NREL Laboratory Directed Research and Development Program. The views expressed in the article do not necessarily represent the views of the DOE or the U.S. Government. The U.S. Government retains and the publisher, by accepting the article for publication, acknowledges that the U.S. Government retains a nonexclusive, paid-up, irrevocable, worldwide license to publish or reproduce the published form of this work, or allow others to do so, for U.S. Government purposes.}
\thanks{G. Cavraro, J. Comden, and A. Bernstein are with the National Renewable Energy Laboratory. Emails: {\tt guido.cavraro@nrel.gov, joshua.comden@nrel.gov, andrey.bernstein@nrel.gov}.
}
}

\begin{document}

\maketitle
\thispagestyle{empty}
\pagestyle{empty}

\begin{abstract}
Energy prices and net power injection limitations regulate the operations in distribution grids and typically ensure that operational constraints are met.
Nevertheless, unexpected or prolonged abnormal events could undermine the grid's functioning.
During contingencies, customers could contribute effectively to sustaining the network by providing services.
This paper proposes an incentive mechanism that promotes users' active participation by essentially altering the energy pricing rule.
The incentives are modeled via a linear function whose parameters can be computed by the system operator (SO) by solving an optimization problem. Feedback-based optimization algorithms are then proposed to seek optimal incentives by leveraging measurements from the grid, even in the case when the SO does not have a full grid and customer information. Numerical simulations on a standard testbed validate the proposed approach.
\end{abstract}

\section{Introduction}\label{sec:intro}

The massive deployment of distributed energy resources (DERs) in distribution networks (DNs) is dramatically changing the electric grid.
Future DNs will be populated by \emph{prosumers} rather than mere users, i.e., entities that can be both producers and consumers of energy~\cite{Cavraro2022Feedback}.
Prosumers could provide services to the grid, e.g., by contributing to voltage profile improvements.
On the other hand, DERs could lead to grid instabilities or damages if not properly managed.
For instance, voltage imbalances due to single-phase rooftop photovoltaic panels in residential, low-voltage DNs have been widely observed~\cite{mohammadi2016challenges}.

Countless methods for regulating DER power outputs have been proposed in the last few years.
Many works in the literature assume that DERs apply power setpoints, possibly directly dispatched from the SO, aiming at the grid's well-being. This would require DER owners to be willing, or forced, to provide grid services even at the expense of their benefit.
DER owners may then have priorities misaligned with those of the SO and refuse to cooperate.
The work~\cite{cui2023load} treated the case in which the DER compliance is modeled with a Bernoulli distribution.
Introducing discomfort functions is a popular way to incorporate human satisfaction in the control or optimization of power systems~\cite{ospina2020personalized}.

SOs could leverage economic incentives like discounts on the energy price to encourage \emph{rational} prosumers, i.e., aiming at maximizing their benefits, to provide grid services~\cite{bhattacharya2022incentive,zhou2017incentive} during abnormal operations, e.g., heat or cold waves~\cite{levin2022extreme}.
Market-based algorithms to incentivize DERs to provide services to the grid while maximizing their objectives and economic benefits were designed in the literature~\cite{mohsenian2010autonomous,maharjan2013dependable,wang2014stackelberg}.
For example, customers may be incentivized to
adjust the output powers of DERs
in real-time to aid voltage regulation~\cite{liu2008distribution}, control the aggregate
network demand~\cite{li2015market}, and follow regulating signals\cite{vrettos2016robust}.
The work~\cite{Alahmed&Cavraro&Bernstein&Tong:23AllertonArXiv} proposes a pricing mechanism for energy communities ensuring that operational constraints are satisfied and guaranteeing that the surplus of each community member is higher than the maximum one under standalone settings.
A trading scheme for increasing the exchange of electricity from prosumers to a distribution network meeting the network constraints is designed in~\cite{azim2023dynamic}.

In this paper, we consider a DN in which prosumers are subject to a net energy metering (NEM 1.0) tariff design, i.e., they incur a linear affine cost for the net power they inject~\cite{alahmed2022net}.
Under contingencies, e.g., heat or cold waves~\cite{levin2022extreme}, the SO encourages prosumers to provide grid services through incentives.
The incentives make rational prosumers change their power demand to support grid operations by essentially altering the energy price, and are designed so that the prosumers are not penalized or rewarded if they do not change their behavior.
The grid model is presented in Section~\ref{sec:modeling}.
Section~\ref{sec:market} reports the considered pricing rule, defines rational customers and introduces the incentive mechanism. The goal of the SO is to design optimal incentive functions that promote the satisfaction of operational constraints while minimizing the cost for the SO. The incentive mechanism is characterized in Section~\ref{sec:inc_prop}. 
Section~\ref{sec:feedback} proposes an asymptotically stable feedback controller for iteratively updating the incentive functions toward their optimum value when the optimization problem cannot be directly solved, e.g., because of a lack of information.
Finally, Section~\ref{sec:simulations} provides numerical results over the standard IEEE 33 bus test feeder, and conclusions are drawn in Section~\ref{sec:conclusions}

\emph{Notation:} Lower- (upper-) case boldface letters denote column vectors (matrices). 
The identity matrix, the vector of all ones, and the vector of all zeros are denoted by $\bI$, $\1$, $\0$; the corresponding dimension will be clear from the context. 
The sets of real numbers and nonnegative real numbers are denoted as $\mathbb R$  and $\mathbb R^+$, respectively. 
The two norm of a matrix $\bA$ is defined by $\|\bA\| = \sqrt{\lambda_{\max}(\bA^\top \bA)}$, where $\lambda_{\max}(\bA^\top \bA)$ is the largest eigenvalue of $\bA^\top \bA$. 


\section{Grid Modeling}
\label{sec:modeling}

A distribution grid with $N+1$ buses can be modeled by an undirected graph $\mcG = (\mcN, \mcE)$, where nodes $\mcN = \{0, 1, \dots, N\}$ are associated with the electrical buses and whose edges represent the electric lines.
The substation is labeled as 0 and is modeled as an ideal voltage generator (the slack bus) imposing the nominal voltage of $1$ p.u. 
%
%
%
%
%
%
%
Each bus, except the substation, is assumed to be a prosumer\cite{Cavraro2022Feedback}. A prosumer is both a producer and a consumer of energy. Prosumer $n$ can generate the active power $r_n \in \realnumbers_+$ potentially exploiting behind-the-meter DERs. Also, prosumer $n$ has an active and reactive power demands $d_n \in \realnumbers_+$ and $q_n \in \realnumbers$. The net active power injection is
\begin{equation}
\label{eq:nodal_inj}
p_n = r_n - d_n.
\end{equation}
Powers will take positive (negative) values, i.e., $p_n, q_n \geq 0$ ($p_n, q_n \leq 0$) when they are \emph{injected into} (\emph{absorbed from}) the grid.
When $p_n \geq 0$, $n$ behaves like a generator; when $p_n \leq 0$, $n$ behaves like a load.
Let $\bd \in \realnumbers^N$ and $\br \in \realnumbers^N$ collect all the demands and DER outputs. Potentially, each prosumer $n$ may have some flexibility in the net power injection, i.e.,
\begin{equation}\label{eq:Conslimit}
p_n \in [\underline{p}_n, \overline{p}_n], \quad n = 1,\dots,N.
\end{equation} 
For loads that are not flexible, e.g., critical loads, we can set trivially $\underline{p}_n = \overline{p}_n$.
The model~\eqref{eq:Conslimit} potentially captures load limitations enforced to not compromise the network's operation, like fixed export limits (e.g., 5kW or 3.5kW \cite{Liu&Ochoa&Riaz&Mancarella&Ting&San&Theunissen:21PEM}) or dynamic operating envelopes~\cite{Alahmed2023Operating}.
The bounds~\eqref{eq:Conslimit} represents essentially an inner approximation of the  power injections feasible region where a solution for the power flow equations is guaranteed to exist and meet the operational constraints.
Hereafter, the generation from DERs is assumed to be constant and $r_n$ is considered a fixed parameter.
The limitation~\eqref{eq:Conslimit} is then equivalent to
\begin{equation}\label{eq:Conslimit1}
d_n \in \mcD_n = [\underline{d}_n, \overline{d}_n], \quad n = 1,\dots,N.
\end{equation} 
Denote by $v_n\in \realnumbers$ the voltage magnitude at bus $n \in \mcN$, and let the vector $\bv\in\realnumbers^N$ collect the voltage magnitudes of buses $1,\ldots,N$.
Voltage magnitudes are nonlinear functions of the power injections; however  first-order Taylor expansion of the power flow equation yields~\cite{Bolognani2015Distributed}
\begin{equation}
\bv = \bR \bp + \bX \bq + \boldsymbol \omega
\label{eq:v=Rp+Xq}
\end{equation}
where $\bR \in \realnumbers_+^{N\times N}$ and $\bX \in \realnumbers_+^{N\times N}$ are symmetric and positive definite matrices~\cite{Cavraro2022Feedback} and $\boldsymbol \omega \in \realnumbers_+^N$.

\section{Energy Market and Incentives}
\label{sec:market}

According to the NEM 1.0, prosumer $n$ net power injection is charged following the rule
$$\gamma(p_n) = - \pi p_n + \pi_0$$
where $\pi > 0$ is the retail rate and $\pi_0$ captures non-volumetric surcharges, e.g., the connection charge~\cite{alahmed2022net}.
When the prosumer net consumes, the first term in $\gamma(p_n)$ is positive, meaning that $n$ is charged. When the prosumer net produces power, the first term in $\gamma(p_n)$ is negative, meaning that $n$ is remunerated.
Without loss of generality, we assume that the coefficients $\pi, \pi_0$ are the same for all the prosumers. 

The {\em surplus} of customer $n$ is the difference between the comfort and the payment from consumption
\begin{align}\label{eq:Surplus_gen}
    \hat S_n(d_n,r_n) & = U_n(d_n) - \gamma(p_n) \notag\\
    & = U_n(d_n) - \pi d_n + \pi r_n - \pi_0
\end{align}
where we used~\eqref{eq:nodal_inj}. The utility of consumption $U_n(d_n)$ is assumed to be strictly concave and continuously differentiable with a marginal utility function $\nabla U_n$.
We denote the inverse marginal utility by ${f}_n:=(\nabla U_n)^{-1}, \forall n\in \mcN$.  

We assume that each prosumer $n$ acts \emph{rationally}, i.e., aims to maximize its surplus; $n$ sets its power demand as the solution of the following optimization problem
\begin{equation}
\label{eq:Sur_max}
\hat{d}_n = \arg \max_{d_n \in \mcD_n}  \hat{S}_n(d_n,r_n)    
\end{equation}
where the price coefficients $\pi$ and $\pi_0$, and the generation $r_n$ are given.
Indeed, the price coefficients are defined in the contract between the utility and the customers and are usually updated once every several months or a few years.
The optimal demand can be easily computed as
\begin{equation}\label{eq:gen_max_d}
    \hat d_n = [f_n(\pi)]_{\mcD_n}.
\end{equation}
where $[\cdot]_{\mcD_n}$ denotes the projection onto the set $\mcD_n$.

Even though~\eqref{eq:Conslimit} typically ensures that the network operates correctly, unexpected or abnormal events, like sudden generation drops or heat and cold waves, might affect the network operations. 
In these cases, the SO could then ask the prosumers to provide grid sevices to avoid grid damages and instabilities.
The SO could compensate prosumer $n$ for their services by means of incentives captured by a function $g_n(d_n,\bxi_n)$, where $\bxi_n$ is a vector of parameters. Define also 
$$\bxi = \begin{bmatrix}
\bxi_1^\top & \dots & \bxi^\top_N    
\end{bmatrix}^\top.$$
The goal of $g_n(d_n,\bxi_n)$ is essentially to shape the surplus~\eqref{eq:Surplus_gen} of prosumer $n$, which becomes
\begin{equation}
S_n(d_n,r_n, \bxi_n) = \hat S_n(d_n,r_n) + g_n(d_n,\bxi_n)
\label{eq:Surplus}
\end{equation}
so that the new solution of
\begin{equation}
\label{eq:Sur_max_g}
d_n^*(r_n,\bxi_n) = \arg \max_{d_n \in \mcD_n} S_n(d_n,r_n)
\end{equation}
is favorable for grid operations, see Figure~\ref{fig:shaping1} and \ref{fig:shaping2}.  Heed that $g_n(\hat d_n, \xi_n)$ equals zero, meaning that no remuneration is given to $n$ if it does not provide services.

\begin{figure}[tb]
    \centering
    \includegraphics[width=0.5\columnwidth]{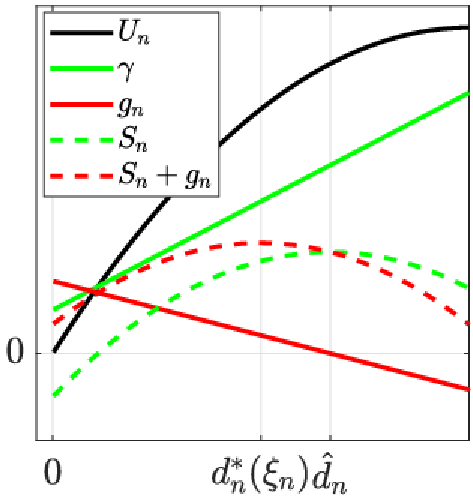}
    \caption{The incentive function shapes prosumer $n$ surplus. Here, the utility of consumption is quadratic and the incentive function is linear. When $\xi_n$ is negative, the demand is reduced.}
    \label{fig:shaping1}
\end{figure}

\begin{figure}[tb]
    \centering
    \includegraphics[width=0.5\columnwidth]{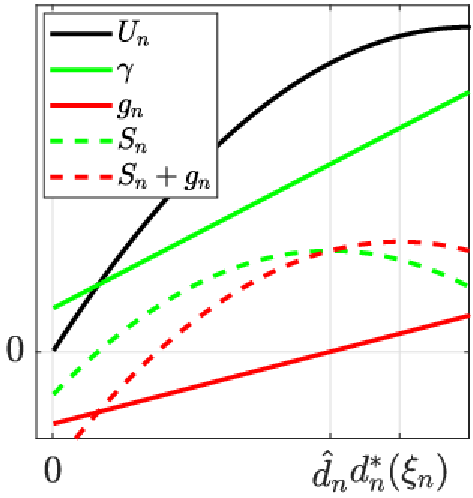}
    \caption{The incentive function shapes prosumer $n$ surplus. Here, the utility of consumption is quadratic and the incentive function is linear. When $\xi_n$ is positive, the demand increases.}
    \label{fig:shaping2}
\end{figure}

The SO's goal is to find the $\bxi_n$'s that minimize the cost of sustaining the distribution grid while ensuring that operational constraints are met, i.e., to solve the optimization problem
\begin{subequations}\label{eq:opt_probl_gen}
\begin{align}
\bxi^* := \arg\min_{\bxi} &\sum_n g_n(d_n^*(r_n,\bxi_n),\bxi_n) - \pi d_n^*(r_n,\bxi_n) + \notag \\
& \qquad - \pi_0 + \pi r_n  \label{eq:optcost}\\
\text{s.t. }& \; \phi(\bd^*,\br,\bxi) = 0 \label{eq:eqconst}\\
& \; \mathbf h(\bd^*,\br,\bxi) \leq 0. \label{eq:inconst}
\end{align}    
\end{subequations}
where the equality constraint~\eqref{eq:eqconst} captures the power flow equations, the inequality~\eqref{eq:inconst} models the grid operational constraints, e.g., voltage or line flows  limits, and we assume that prosumers behave rationally and aims to maximize their profit.
The interactions between the utility company (leader) and prosumers (followers) is a Stackelberg game~\cite{YU2016702}, where the players select the optimal strategy by solving the optimization problems~\eqref{eq:Sur_max_g} and~\eqref{eq:opt_probl_gen}.


\section{An Incentive Mechanism for Grid Services}
\label{sec:inc_prop}

While problems~\eqref{eq:Sur_max_g} and~\eqref{eq:opt_probl_gen} are stated in general forms, here we will study the properties of the proposed framework under particular choices of functions and parameters that will make the results meaningful and interpretable. Nevertheless, all the methods described in the following, after proper modifications, could still be used to solve~\eqref{eq:Sur_max_g} and~\eqref{eq:opt_probl_gen}.

\paragraph{Quadratic utility functions}
similar to what is commonly done in the literature, e.g., see~\cite{YU2016702,samadi2010optimal}, we consider quadratic prosumer utility functions
$$U_n(d_n) = - \frac{\alpha_n}{2} d_n^2 + \beta_n d_n, \quad \alpha_n \in \realnumbers^+,\beta_n \geq \pi.$$
\paragraph{Arbitrary power demand}
we assume that the prosumers can choose the $d_n$'s to be an arbitrary nonnegative number, i.e., we disregard~\eqref{eq:Conslimit}.
This choice slightly simplifies the form of the prosumer demand~\eqref{eq:gen_max_d}; indeed, we can write the \emph{nominal} (i.e., the one in the absence of incentives) demand and the net power injection for prosumer $n$ as
$$\hat d_n := \frac{\beta_n-\pi}{\alpha_n}, \qquad \hat p_n := r_n - \hat d_n.$$
The projection in~\eqref{eq:gen_max_d} would just complicate the notation hereafter without adding anything conceptually and can be performed easily in practical applications. 

\paragraph{Linear incentive functions}
we consider incentive functions of the form
\begin{equation*}
\rho_n(d_n,\xi_n,\eta_n) = \xi_n d_n + \eta_n
\end{equation*}
where $\xi_n, \eta_n \in \realnumbers$. 
To avoid a prosumer being charged or remunerated even if it does not change its power demand, we impose that 
$$\rho_n(\hat d_n) = 0.$$
This yields
$\eta_n = - \xi_n \hat d_n.$

Hence, the incentive function becomes
\begin{align}
\label{eq:incentive_g}
g_n(d_n,\xi_n) & = \rho_n(d_n,\xi_n,-\xi_n \hat d_n) = \xi_n (d_n - \hat d_n)
\end{align}
from which it is clear that a prosumer will be remunerated if it deviates from his nominal consumption $\hat d_n$. 
The parameters vector $\bxi \in \realnumbers^N$ hereafter collects all the $\xi_n$'s. 

\paragraph{Approximated power exchange} We neglect the power losses and we assume that the power delivered to the distribution network through the substation is 
\begin{equation*}
p_0 = - \sum_n p_n = \1^\top(\bd - \br).    
\end{equation*}
Together with the approximation~\eqref{eq:v=Rp+Xq}, the former equation yields the convex optimization problem reported in the following. 
The SO could however in principle solve an optimization problem considering the true power flow equations. 

The former choices yield the next quantities. The surplus~\eqref{eq:Surplus} can be written
\begin{align*}
S_n(d_n,\xi_n) & = - \frac{\alpha_n}{2} d_n^2 + \beta_n d_n - \pi d_n + \pi r_n \\
& \qquad  - \pi_0 + \xi_n (d_n - \hat d_n)
\end{align*}
and the optimal consumption~\eqref{eq:gen_max_d} for prosumer $n$ is
\begin{align*}
d_n^*(\xi_n) &=\frac{\beta_n - \pi + \xi_n}{\alpha_n} = \hat d_n + \frac{ \xi_n}{\alpha_n}  
\end{align*}
that is, the new surplus maximizer is a linear perturbation of the one  without incentives.
Let us collect all the optimal consumptions with or without incentives in the vectors $\bd^*(\bxi)$ and $\hat \bd$, respectively.
Then,
\begin{align}
\label{eq:hat_d}
&\hat \bd = \bA \bbeta - \pi \bA \1\\
\label{eq:opt_d}
&\bd^*(\bxi) = \hat \bd + \bA \bxi
\end{align}
with
$\bA := \diag(\ba), \quad
\ba := \begin{bmatrix}
\frac{1}{\alpha_1} & \dots & \frac{1}{\alpha_N} 
\end{bmatrix} ^\top.$

To avoid pathological situations in which the incentive mechanism yields negative power demand, add the condition
\begin{equation}
\label{eq:xi_const}
\bxi \geq \pi \1 - \bbeta, \qquad \bbeta := 
\begin{bmatrix}
\beta_1 & \beta_2 & \dots & \beta_N    
\end{bmatrix}^\top.
\end{equation}
Under the consumption $\bd^*(\bxi)$, the power delivered to the distribution grid is 
\begin{equation}
\label{eq:p0_xi}
p_0(\bxi) = \1^\top \bA \bxi + \1^\top \hat \bd - \1^\top \br.
\end{equation}
Also, the remuneration due to the prosumers, i.e., the sum of the incentives, for their services is quadratic in $\bxi$. Indeed:
\begin{align*}
\sum_n & \, g_n(d_n^*(\xi_n),\xi_n) - \pi d_n^*(\xi_n) + \pi r_n - \pi_0 = \\
& = \sum_n \frac{\xi_n^2}{\alpha_n} - \frac{\pi}{\alpha_n} \xi_n - \gamma(\hat p_n)  = \bxi^\top \bA \bxi + \bb^\top \bxi + c
\end{align*}
where 
$ \bb := - \pi \bA \1,  \quad c := - \sum_n \gamma(\hat p_n).
$

Using~\eqref{eq:v=Rp+Xq} and~\eqref{eq:opt_d}, the voltage magnitudes are expressed as a function of the incentive
\begin{align}
\label{eq:vmagn}
\bv = - \bR \bA\bxi + \hat \bv
\end{align}
with 
$\hat \bv := \bR \br - \bR \hat \bd + \bX \bq + \boldsymbol \omega$.

The SO's task is to design the incentive mechanism in order to keep voltages and the power supplied to the grid within prescribed limits. That is, we consider the following particular instance of~\eqref{eq:opt_probl_gen}
\begin{subequations}
\label{eq:opt_prob_part}
\begin{align}
\bxi^* = \arg \min_{\bxi}& \; \bxi^\top \bA \bxi + \bb^\top \bxi + c \label{eq:opt_prob_part_cost}\\
\text{s.t. }& \; \eqref{eq:xi_const} - \eqref{eq:p0_xi} - \eqref{eq:vmagn} \notag\\
& \; \underline \bv \leq \bv \leq \overline \bv \label{eq:opt_prob_part_volt} \\
& \; \underline p_0 \leq p_0 \leq \overline p_0 \label{eq:opt_prob_part_PCC} 
\end{align}
\end{subequations}
The constraint~\eqref{eq:opt_prob_part_volt} captures voltage operational constraints; whereas~\eqref{eq:opt_prob_part_PCC} enforces the power exchange with the external network to be within a desired interval, possibly modeling the case in which the grid is required to behave as a Virtual Power Plant.
We will also assume that the feasible set described by equations~\eqref{eq:xi_const}, \eqref{eq:p0_xi}, \eqref{eq:vmagn}, \eqref{eq:opt_prob_part_volt}, and~\eqref{eq:opt_prob_part_PCC} is \emph{non empty}.
The objective~\eqref{eq:opt_prob_part_cost} is convex. 
Hence, problem~\eqref{eq:opt_prob_part} is strictly convex and admits a unique minimizer. 


\section{Achieving the Optimal Incentives}
\label{sec:feedback}

The optimal incentive $\bxi^*$ can be computed by directly solving problem~\eqref{eq:opt_prob_part} when the SO has full network information, i.e., it knows the grid parameters $\bR$, the power demands $\bd$ and $\bq$, the DER power outputs $\br$, and the user preferences parameters $\bA$ and $\bbeta$. 
However, such a scenario of perfect grid information is unusual in distribution networks, for instance, because of a lack of real-time metering infrastructure. 
Hence, we propose the following feedback optimization algorithms in which the missing information is compensated by \emph{measurements} and problem~\eqref{eq:opt_prob_part} is  solved iteratively.

\subsection{A Dual Ascent Method for the Full Information Case}
Suppose the SO does not have available real-time information about behind-the-meter generation $r_n$ and reactive power demand $q_n$ of prosumers $1,\dots,N$.
A dual ascent approach can be adopted. Introduce the Lagrangian
\begin{align}
\mcL(\bxi, & \overline \blambda,\underline \blambda,\bnu, \overline \mu,\underline \mu) = \bxi^\top \bA \bxi + \bb^\top \bxi + c  \notag \\
& + \overline \blambda^\top(\bv - \overline \bv) - \underline \blambda^\top (\bv - \underline \bv) + \bnu^\top (\pi \1 - \bbeta - \bxi) \notag \\
& + \overline \mu (p_0 - \overline p_0) - \underline \mu (p_0 - \underline p_0). \label{eq:Lagrangian}
\end{align}
The dual ascent algorithm solving~\eqref{eq:opt_prob_part} reads 
\begin{subequations}
\label{eq:DA}
\begin{align}
& \bxi(t+1) = \arg \min_{\bxi} \mcL(\bxi(t), \overline \blambda(t),\underline \blambda(t), \bnu(t), \overline \mu(t),\underline \mu(t)) \notag\\
&  = \frac 1 2\Big(\bR(\overline \blambda(t)) - \underline \blambda(t))  + \1 (\underline \mu(t) - \overline \mu(t)) + \bA^{-1} \bnu(t) + \pi \1  \Big)\label{eq:DU_1}\\
& \overline \blambda(t+1) = \big[ \blambda(t) + \epsilon(\bv(t) - \overline \bv) \big]_{\realnumbers^N_+} \label{eq:DU_2}\\
& \underline \blambda(t+1) = \big[ \blambda(t) + \epsilon(\underline \bv - \bv(t)) \big]_{\realnumbers^N_+}   \label{eq:DU_3}\\
& \underline \mu(t+1) = \big[ \mu(t) + \epsilon(\underline p_0 - p_0(t)) \big]_{\realnumbers_+} \label{eq:DU_4}\\
& \overline \mu(t+1) = \big[ \mu(t) + \epsilon(p_0(t) - \overline p_0) \big]_{\realnumbers_+}  \label{eq:DU_5} \\
& \bnu(t+1) = \big[\bnu(t) + \epsilon (\pi \1 - \bbeta - \bxi(t)) \big]_{\realnumbers^N_+}  \label{eq:DU_6}
\end{align}
\end{subequations}
Algorithm~\eqref{eq:DA} has a feedback control implementation. 
While the minimization in~\eqref{eq:DU_1} can be readily performed by the SO, 
voltage and power measurements enable the Lagrange multipliers updates~\eqref{eq:DU_2}--\eqref{eq:DU_5}. The incentive parameters are iteratively updated in~\eqref{eq:DA} until convergence to their optimum values. 
A condition for the convergence of algorithm~\eqref{eq:DA} is provided in the next result.
\begin{proposition}
\label{prop:DA_stability}
Consider the dual ascent control scheme~\eqref{eq:DA} 
and define the matrix
$$\bPhi:= 
\begin{bmatrix}
-\bA \bR & \bA \bR & \bA \1 & -\bA \1 & -\bI
\end{bmatrix}^\top \in \realnumbers^{(3N + 2) \times N}.$$
Then~\eqref{eq:DA} converges to the unique minimizer of~\eqref{eq:opt_prob_part} if
\begin{equation}
\label{eq:epsilon}
\epsilon \leq \frac{4}{\| \bPhi^\top \bA^{-1} \bPhi \|}.
\end{equation}
\end{proposition}

\begin{proof}
After defining the vector $\bphi \in \realnumbers^{3N +2}$

\begin{small}
$$\bphi=
\begin{bmatrix}
(\hat \bv - \overline \bv)^\top & (\underline \bv - \hat \bv)^\top & \hat \bd^\top \1 - \overline p_0 & \underline p_0 - \hat \bd^\top \1 &  \pi \1^\top - \bbeta  
\end{bmatrix}^\top$$    
\end{small}

\noindent problem~\eqref{eq:opt_prob_part} can be rewritten as 
\begin{subequations}
\label{eq:opt_comp}
\begin{align}
\bxi^* = \arg \min_{\bxi}& \; \bxi^\top \bA \bxi + \bb^\top \bxi + c \label{eq:opt_comp_cost}\\
\text{s.t. }& \; \bPhi \bxi + \bphi \leq \0 \label{eq:opt_comp_const} 
\end{align}
\end{subequations}
Collecting the Lagrange multipliers of~\eqref{eq:Lagrangian} in the vector 
$$\btheta := 
\begin{bmatrix}
\overline{\blambda}^\top & \underline \blambda^\top & \overline{\mu}^\top & \underline \mu^\top & \bnu^\top
\end{bmatrix}^\top, \quad \btheta \in \realnumbers_+^{3N+2}$$
the Lagrangian of~\eqref{eq:opt_comp} is
$$
\mcL(\bxi,\btheta) = \bxi^\top \bA \bxi + \bb^\top \bxi + c + \btheta^\top (\bPhi \bxi + \bphi) 
$$
whose minimizer, w.r.t. the primal variable, is
$$\bxi(\btheta) = - \frac{\bA^{-1}}{2}(\bb + \bPhi \btheta)$$
which can be used to obtain the dual problem %
$$
\max_{\btheta \in \realnumbers_+^{3N+2}} \bh(\btheta)
$$
where 
\begin{equation*}
\bh(\btheta) = \btheta^\top \Big(\bphi - \frac{\bPhi \bA^{-1} \bb}{2} \Big) -\btheta^\top \frac{\bPhi \bA^{-1}\bPhi^\top }{4}\btheta - \frac{\bb^\top \bA^{-1} \bb}{4}.
\end{equation*}
Since the former is a quadratic problem with a nonempty feasible set and linear constraints, it features zero duality gap with~\eqref{eq:opt_prob_part} because the Slater's conditions hold true~\cite{boyd2004convex}. 
The gradient of $\bh$ is given by
\begin{equation*}
\nabla \bh(\btheta) = \frac{\bPhi \bA^{-1}\bPhi^\top }{2}\btheta - \Big(\frac{\bPhi \bA^{-1} \bb}{2} - \bphi \Big) 
\end{equation*}
and the dual ascent algorithm~\eqref{eq:DA} becomes
$$\btheta(t+1) = \bff(\btheta(t))$$
where
$$
\bff(\btheta) = \big[\btheta + \epsilon \nabla \bh(\btheta) \big]_{\realnumbers_+^{3N+2}}
$$
By recalling that the projection is a nonexpansive operator, we can prove that, under condition~\eqref{eq:epsilon}, the map $\bff$ is a contraction.
Indeed,
\begin{align*}
&\|\bff(\btheta) - \bff(\btheta')\| = \\
& = \big \|\big[\btheta + \epsilon \nabla \bh(\btheta) \big]_{\realnumbers_+^{3N+2}} -
\big[\btheta' + \epsilon \nabla \bh(\btheta') \big]_{\realnumbers_+^{3N+2}} \big\| \\
& \leq \|(\btheta + \epsilon \nabla \bh(\btheta) ) -
(\btheta' + \epsilon \nabla \bh(\btheta')) \| \\
&\leq \Big\| \Big(\bI - \epsilon \frac{\bPhi \bA^{-1}\bPhi^\top }{2} \Big)(\btheta - \btheta') \Big\| \\
& \leq \Big\| \bI - \epsilon \frac{\bPhi \bA^{-1}\bPhi^\top }{2} \Big\| \|\btheta - \btheta'\|  \leq \|\btheta - \btheta'\|
\end{align*}
The asymptotic stability then follows.
\end{proof}

\subsection{First-Order and Zero-Order Feedback-Based Methods for the Incomplete Information Case}
Consider the case where the SO does not have access to full information about the grid state and/or the surplus model of the prosumers. To design algorithms that are suitable for this limited information case, consider the implicit form of the Lagrangian \eqref{eq:Lagrangian} as a function of the optimal demand:
\begin{align}
\mcL(\bxi, & \overline \blambda,\underline \blambda, \bnu, \overline \mu,\underline \mu) = \bxi^\top (\bd^*(\bxi) - \hat \bd )  - \pi \1^\top \bd^*(\bxi) + c' \notag \\
& + \overline \blambda^\top(\bv - \overline \bv) - \underline \blambda^\top (\bv - \underline \bv) - \bnu^\top \bd^*(\bxi) \notag \\
& + \overline \mu (p_0 - \overline p_0) - \underline \mu (p_0 - \underline p_0). \label{eq:Lagrangian_gen}
\end{align}
where $c' = \sum_n (\pi r_n - \pi_0)$.
Following the literature on feedback-based optimization \cite{Bolognani2015Distributed,colombino2019online}, we then propose two methods to seek the saddle points of \eqref{eq:Lagrangian_gen} by leveraging measurements; the convergence analysis\footnote{We note that the proposed algorithms are applicable to any (not necessarily linear) optimal demand function $\bd^*(\bxi)$. If the optimal demand is linear, then the convergence to the unique optimum of \eqref{eq:opt_prob_part} is guaranteed under the suitable choice of algorithm parameters \cite{colombino2019online,chen2020}.} of these methods is left for future work.

The first-order primal-dual method for \eqref{eq:Lagrangian_gen} reads
\begin{subequations}
\label{eq:first_order_feedback}
\begin{align}
& \bxi(t+1) = \bxi(t) - \epsilon \Big(
\bd^*(\bxi(t)) - \hat \bd + \nabla \bd^*(\bxi(t)) \bxi(t) 
\notag\\
&\qquad \quad - \pi \nabla \bd^*(\bxi(t)) \1 + \nabla \bv (t) \left(\overline \blambda(t) - \underline \blambda(t)\right) \notag \\
& \qquad \quad  + \nabla p_0 (t) \left(\overline \mu(t) - \underline \mu(t)   \right) \, \, - \nabla \bd^*(\bxi(t)) \bnu(t)
\Big)
\label{eq:DU_1_gen_fo}\\
& \eqref{eq:DU_2}-\eqref{eq:DU_6}
\end{align}
\end{subequations}
Note that to implement the primal update \eqref{eq:DU_1_gen_fo}, one needs the measurement of the demand $\bd(t) = \bd^*(\bxi(t))$; the sensitivity of the demand to the incentive signal given by the gradient matrix $\nabla(t) = \nabla \bd^*(\bxi(t))$ (cf.~matrix $\bA^\top$ in \eqref{eq:opt_d}); and the sensitivity of the power-flow model to the incentives given by the  gradients  $\nabla \bv(t)$ and $\nabla p_0(t)$ (cf.~matrix $(\bR \bA)^\top$ in \eqref{eq:vmagn}).  These matrices can be estimated from historical data (e.g., from previous demand response events).

In the most extreme case when also the sensitivity matrices above are unknown, we propose to use a zero-order method to seek saddle points of \eqref{eq:Lagrangian_gen} similar to, e.g., \cite{chen2020}. In particular, we employ a double-evaluation approach for approximating the gradient of the Lagrangian:
\begin{align} 
    & \widehat{\nabla} \mcL (t) := \frac{\bzeta(t)}{2\sigma} \Big [ \mcL \left(\widehat{\bxi}_{+} (t), \btheta(t) \right) -  \mcL \left(\widehat{\bxi}_-(t), \btheta(t) \right)  \Big] \label{eq:Lagr_zo}
\end{align}
where \emph{perturbed} incentives $\widehat{\bxi}_+(t)$ and $\widehat{\bxi}_{-}(t)$ are applied to the system with
$\widehat{\bxi}_{\pm}(t) := \bxi(t) \pm \sigma \bzeta(t).$ Here, $\sigma > 0$ is a parameter that controls the magnitude of perturbation, and $\bzeta(t) \in \realnumbers^N$ is a perturbation signal which can be either chosen as a random or deterministic process. In Section \ref{sec:simulations}, we show an application in which $\bzeta(t)$ is a random signal.

With approximation \eqref{eq:Lagr_zo} at hand, the primal update \eqref{eq:DU_1_gen_fo} is replaced with 
\begin{align} \label{eq:primal_zo}
& \bxi(t+1) = \bxi(t) - \epsilon \widehat{\nabla} \mcL (t).
\end{align}
Observe that \eqref{eq:primal_zo} can be implemented in a complete model-free fashion provided that the measurements of demand, voltages, and aggregate power are available.

\section{Numerical Illustration}
\label{sec:simulations}

Here, we validate the incentive mechanism and the feedback-based optimization algorithms from Section \ref{sec:feedback} 
on a realistic distribution feeder.
The IEEE 33-bus radial distribution network~\cite{baran1989network} was simulated using PandaPower with 32 loads chosen randomly from 114 apartments sourced from the UMass Trace Repository~\cite{barker2012smart} to be placed at each of the 32 load nodes.
A normalized retail price $\pi$ was set at 1.0 and the prosumer quadratic utility function coefficients $\alpha_n$ were chosen uniformly at random between 0.3 and 3.0 so that the most responsive prosumer would be a magnitude more responsive than the least responsive prosumer.

We chose a time instance with a heavy enough load that would cause some of the nodal voltage magnitudes to drop below their 0.95 p.u. lower bound.
To counteract the heavy load and keep the voltages within bounds, a generator was placed on bus 31 with a capacity of 6 times its default node load size.
Virtual power plant bounds of $\pm 0.2$ MW were placed around the power going into the feeder with the generator on at full capacity.
The evaluation of the algorithms starts when the generator is turned off and the incentive mechanism is switched on simultaneously.
The parameters of the algorithms were tuned to the following: $\epsilon=0.5$ for the dual ascent; $\epsilon=0.3$ for the first-order algorithm; $\sigma=0.02$ and $\epsilon=0.05$ for the zero-order method.
A vector of uniform random variables between -1 and 1 was chosen for $\boldsymbol{\zeta}(t)$.

The algorithms are compared by showing their total incentive, minimum nodal voltage magnitude, and feeder power versus the number of iterations, in Figures \ref{fig:incentive_total}, \ref{fig:voltage_min}, and \ref{fig:vpp_P0}, respectively.
As expected, the more information we have about the prosumers, the faster we can approach an optimal $\boldsymbol{\xi}$.
Dual ascent utilizes complete knowledge of the prosumer utility functions to converge the fastest, while the first-order algorithm utilizes only the prosumer sensitivities to incentives to converge at a slightly slower rate.  However, the zero-order algorithm has no knowledge of the prosumers and requires exploration to slowly find effective values of $\boldsymbol{\xi}$ with respect to the voltage and virtual power plant bounds.

\begin{figure}[tb]
    \centering
    \includegraphics[width=0.98\columnwidth]{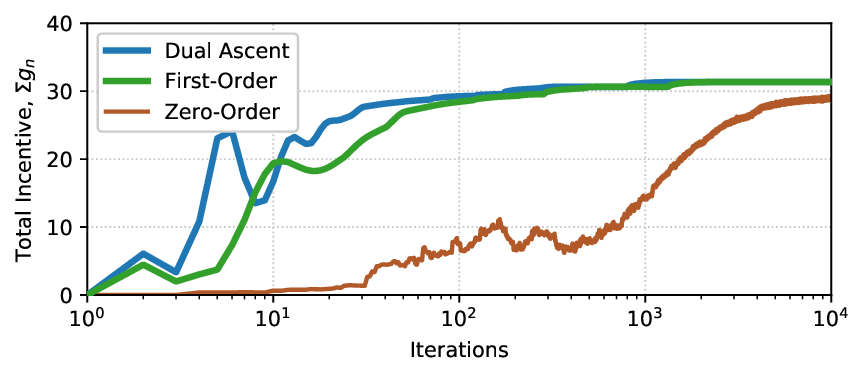}
    \caption{Total incentive of customers vs. the number of iterations.}
    \label{fig:incentive_total}
\end{figure}

\begin{figure}[tb]
    \centering
    \includegraphics[width=0.98\columnwidth]{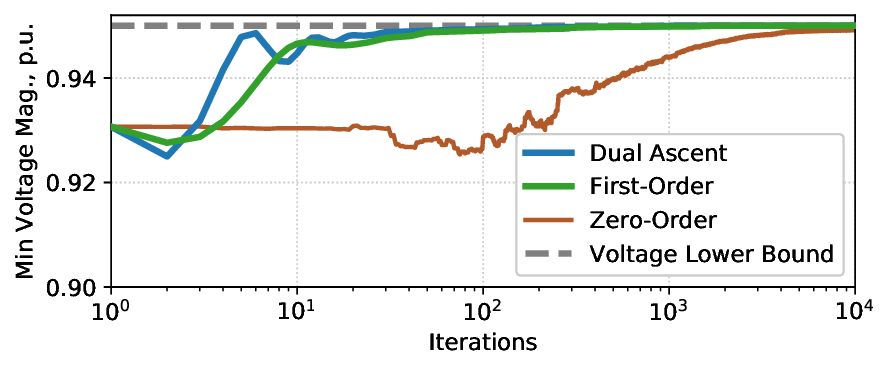}
    \caption{Minimum nodal voltage magnitude vs. the number of iterations.}
    \label{fig:voltage_min}
\end{figure}

\begin{figure}[tb]
    \centering
    \includegraphics[width=0.98\columnwidth]{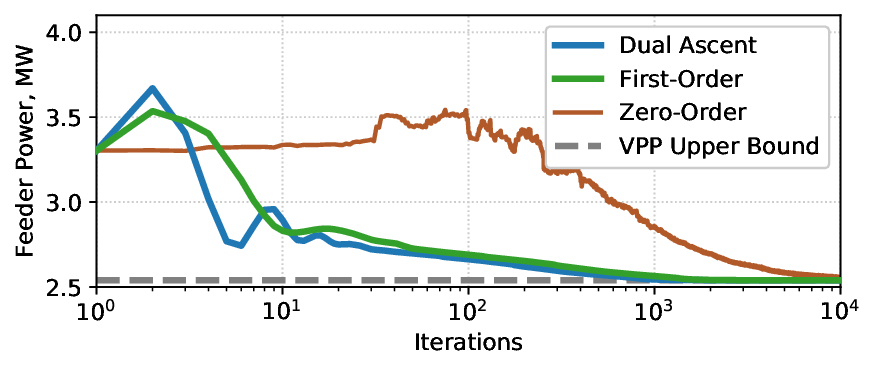}
    \caption{Feeder power vs. the number of iterations.}
    \label{fig:vpp_P0}
\end{figure}

\section{Conclusion}
\label{sec:conclusions}

We have presented an incentive mechanism that, by essentially changing the energy price, makes rational users change their demand and provide grid services, e.g., voltage and power regulation.
The incentives are described here with affine functions. The function parameters that achieve the desired grid performance while minimizing the overall cost for the SO can be computed by solving an optimization problem.
For the case in which the problem cannot be directly solved, because some grid/customer information is not available to the SO, we devised feedback control algorithms that iteratively update the incentives until convergence to the optimum. Future research directions include relaxing the assumption on constant DER power production and considering nonlinear incentive functions.



\bibliographystyle{IEEEtran}
\bibliography{myabrv,Bibliography}

\end{document}